\documentclass{article}
\usepackage{graphicx} 
\usepackage[a4paper, margin=2.5cm]{geometry}
\usepackage{amsthm}
\usepackage{authblk}

\theoremstyle{plain} 
\newtheorem{theorem}{Theorem}
\newtheorem{lemma}{Lemma}

\theoremstyle{definition} 

\theoremstyle{remark} 

\graphicspath{{figures/}}

\usepackage{xcolor}

\title{On the complexity of covering points by guillotine cuts}

\author[1]{Delia Garijo}
\affil[1]{Universidad de Sevilla, Sevilla, Spain} 
\author[1]{Alberto Márquez}
\author[2]{Rodrigo I. Silveira}
\affil[2]{Universitat Politècnica de Catalunya, Barcelona, Spain}

\date{}

\begin{document}

\maketitle

\begin{abstract}
    We show that the problem of covering a set of points in the plane with a minimum number of guillotine cuts is NP-complete.
    To that end, first we present a new NP-completeness proof for the problem of covering points with disjoint line segments.
    Then, we adapt the proof to show that the problem remains NP-complete when the segments are guillotine cuts.
\end{abstract}

\section{Introduction}
Covering problems are among the most fundamental topics in combinatorial optimization and algorithms. 
These problems include classic combinatorial questions such as the \textit{set cover problem}, which asks for a minimum collection of sets to cover all elements in a universe, and the \textit{vertex cover problem} in graphs, where the goal is to compute a minimum set of vertices, such that every edge is incident to at least one vertex in the set.

These foundational problems have plenty of  natural variants that take place in a geometric setting. 
In geometric covering problems, the items to be covered are usually points or regions in some geometric space, and the covering elements are specific geometric shapes like lines, disks, or squares. 
The inherent constraints of the geometric arrangements—unlike their purely combinatorial predecessors—add structure to the problems, but rarely simplify their computational complexity.
This means that, in general,  geometric covering problems remain NP-hard.

Geometric problems that have been shown to be NP-hard range from classic formulations concerned with covering a set of points in the plane with the smallest possible number of lines~\cite{MegiddoT82}, unit squares~\cite{FowlerPT81} or unit disks~\cite{HochbaumM85}, to more specialized settings as covering bichromatic point sets with monochromatic disks~\cite{CabelloDP13} or covering the interior of a polygon with point guards--- which leads to the rich variety of problems known as \emph{art gallery problems} (see~\cite{LeeLin1986} for a classic paper on this topic).

In this note, we focus on two variants of the problem of covering points with lines.
In the classic \emph{point line cover} problem, one is given $n$ points in the plane, with coordinates $(x_1, y_1), \dots, (x_n, y_n)$, and the goal is to find a set of lines $\ell_1, \dots, \ell_r$ of minimum cardinality such that for each $1\leq i \leq n$, point $(x_i, y_i)$ lies on at least one line $\ell_j$, for some $1 \leq j \leq r$.
As mentioned above, the point line cover problem has been known to  be NP-hard for decades~\cite{MegiddoT82}.

Our interest is on variants of the point line cover problem where the objects used to cover points are \emph{disjoint}.
Perhaps the most natural variant in this direction is covering with \textit{disjoint line segments}.
Surprisingly, the complexity of this problem remained open until recently, when Grelier~\cite{grelier:LIPIcs.SoCG.2022.45} proved that the problem is NP-hard.

The main goal of this work is to prove a similar result for the problem of covering points with \textit{guillotine cuts}.
Informally, a guillotine cut is a straight line segment (possibly degenerating into a halfline or line) that extends from one segment  to another segment.
Guillotine cuts (also known as guillotine partitions) are subdivisions formed by recursively applying guillotine cuts.
They appear naturally in cutting stock problems, and the resulting segments are disjoint, except at endpoints.
See Figure~\ref{fig:example-problems} for an illustration of both problems.

While the problem of covering points by disjoint segments is known NP-hard, the reduction used by Grelier~\cite{grelier:LIPIcs.SoCG.2022.45} is from the problem Maximum Independent Set in Intersection Graphs
of Segments.
Since it is not evident how to adapt such a reduction for the particular case of segments that are guillotine cuts, in this work we present a new proof of NP-hardness for  covering points by disjoint segments, reducing from planar monotone 3-SAT, which we can later adapt to work with guillotine cuts.

While verifying that these problems are NP-hard is not very surprising, confirming hardness is an important prerequisite for any further algorithmic study of the problems.

In the following, we will consider the slightly more general \textit{decision versions} of the problems above.
Namely, we will be interested in determining, given $n$ points and a value $k$, if the given points can be covered by $k$ disjoint line segments, or $k$ guillotine cuts.


\begin{figure}[tb]
    \centering
    \includegraphics{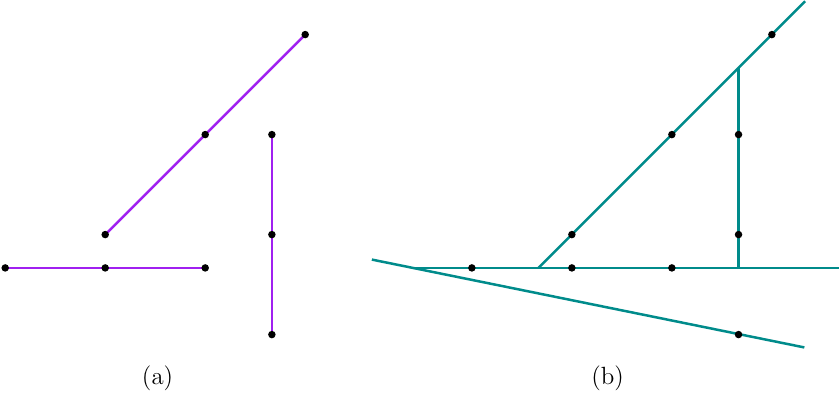}
    \caption{Optimal coverings for $n=9$ points with (a) three disjoint segments, (b) four guillotine cuts.}
    \label{fig:example-problems}
\end{figure}

\section{Covering points by disjoint segments}

In this section present a new NP-completeness proof for the problem of covering points with disjoint line segments, using a construction that we can later adapt to guillotine cuts.

We prove that, given a set of $n$ points in the plane and a value $k$, determining if all points can be covered by $k$ disjoint line segments is NP-complete. 
Our proof is an adaptation of the original proof by Megiddo and Tamir showing that the point line cover problem is NP-complete~\cite{MegiddoT82}, modified to use disjoint line segments instead of arbitrary lines.

The proof by Megiddo and Tamir consists of a reduction from the 3-satisfiability problem (3-SAT) to point cover.
The idea is that one is given an instance of 3-SAT, and based on it one builds a set of points, such that a solution to the decision version of the point cover problem on those points can be translated into a solution to the 3-SAT problem.

More precisely, the input to 3-SAT consists of a boolean  formula in conjunctive normal form, i.e., a conjunction of clauses  $C_1 \land \dots \land C_m$, where $C_j$ is a clause of the form $(l_i \lor l_j \lor l_k)$, with each of $l_i$ being a literal (variable or negated variable) from the variable set $x_1,  \dots, x_n$.
Thus the input has $m$ clauses and $n$ variables.
We assume that each clause contains at least two literals.

Megiddo and Tamir show how to construct from the formula a set of $m+nm^2$ points, such that the input formula is satisfiable if and only if the set of points can be covered by $nm$ lines.
Unfortunately, it is unclear how to shorten or adapt the set of lines needed to optimally cover the points in their reduction to become an equivalent set of  \textit{disjoint} line segments.
Therefore, we cannot apply their reduction directly.

In this section, we present an adaptation of the proof to use disjoint line segments.
Essentially, we redefine the construction of the original proof so that all lines in an optimal solution to line cover can be shortened to segments, so that no two of them intersect.

To that end, we introduce two main changes:
(1) we change the problem of the reduction from 3-SAT to \emph{planar monotone 3-SAT} (RPM3-SAT), proved NP-complete by de Berg and Khosravi \cite{de2010optimal};
(2) we invert the overall structure of the construction, placing the variables---as opposed to the clauses---along a central horizontal line.

The structure of the construction is designed  to work also for guillotine cuts, as we show in the next section.

In RPM3-SAT, the given formula is \emph{monotone}, i.e., each clause has only positive or only negative literals. 
Moreover, the formula is given together with a drawing of the graph $G$ of the formula.
The drawing of $G$ satisfies the following convenient properties:
(i) Each clause and each variable is drawn as an axis-aligned rectangles,
(ii) all variable-rectangles are lying on a horizontal line,
(iii) positive clauses are drawn above the horizontal line, and negative clauses below,
(iv) the edges connecting the variables to the clauses are vertical segments, and
(v) the drawing has no crossings.
See Figure \ref{fig:monotone-3SAT-example} for an example.

\begin{figure}
    \centering
    \includegraphics{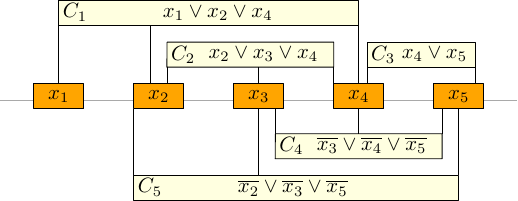}
    \caption{Example of the drawing of a  monotone 3-SAT formula with five variables and five clauses.}
    \label{fig:monotone-3SAT-example}
\end{figure}

Without loss of generality, in the following we assume that in the input formula all clauses are different, and that no clause is contained in another clause (e.g., we cannot have a clause $(x_1 \lor x_2)$ and another clause $(x_1 \lor x_2 \lor x_3)$). 

\subsection{Adapting the drawing}
For our purposes, we transform the standard drawing of $G$, described above, into a different drawing, as follows.
Refer to Figure \ref{fig:instance-example2}.
\begin{figure}
    \centering
    \includegraphics{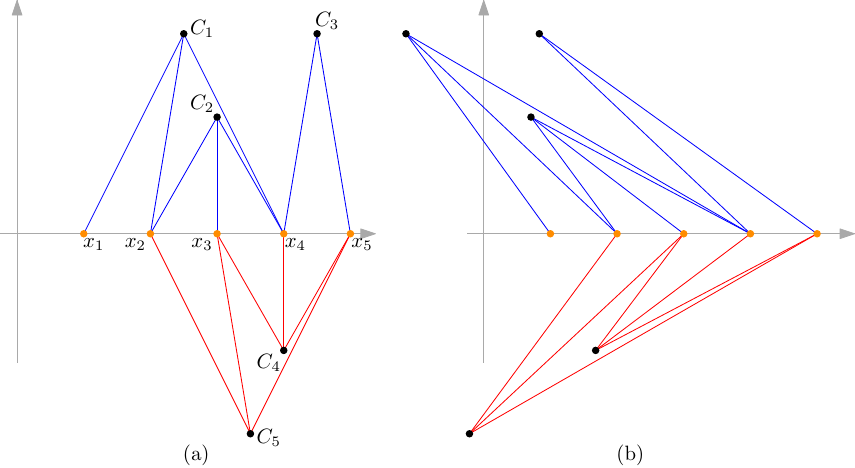}
    \caption{Transformed drawing of example in Figure~\ref{fig:monotone-3SAT-example} after (a) first phase and (b) second phase.}
    \label{fig:instance-example2}
\end{figure}
We present the transformation in two phases.
\paragraph{Phase 1}
First, each variable becomes represented by one \textit{variable point} on the positive $x$-axis.
Second, each clause becomes represented by a \textit{clause point} and two or three line segments connecting that point to the point of each variable.

Next we make this precise; refer to Figure \ref{fig:instance-example2}(a).
To create variable points we follow the same variable order as in the input drawing, assigning to the $i$th variable $x_i$ (from left to right), point $(i,0)$. 
These points will be called \emph{variable points} (orange in the figure).

For the set of positive clauses, assumed to be drawn above the $x$-axis, we proceed as follows.
Let $h$ be the nesting level of a clause, where $h=0$ means that, in the drawing, the clause is not nested into any other clause; $h=1$ means the clause is contained in exactly one clause, and so on. 
For a clause of level $h$ involving variables with indices $i<j<k$, we define the \textit{clause point} as follows.
The $x$-coordinate is simply the midpoint $(i+k)/2$. 
The $y$-coordinate, which is the same  for all clauses of level $h$, is given by  $\max_{C_{abc}} (c-a)- \varepsilon h$, where the maximum is taken over all clauses of level $h$ involving any variables $a,b,c$ ($a<b<c$), and $\varepsilon$ is a small enough value to avoid collinearities, e.g., $\varepsilon = \frac{1}{m^2}$.

Each clause point is connected with a line segment to the two or three variable points that it includes.
For negative clauses we apply a symmetric construction.

Note that, by construction, the new drawing preserves the clause nesting of the original drawing, and all resulting segments are interior disjoint.

\paragraph{Phase 2} Next we transform the previous drawing once again, to ensure that all segments incident to variables coming from above (positive clauses) have  negative slopes, and all segments coming from below (negative clauses) have  positive slopes.
Refer to Figure \ref{fig:instance-example2}(b).

To this end, for the positive clauses (above the $x$-axis) we apply a shear transformation that maps a vector of slope $1$ to one of negative slope, say, slope $-4$. (e.g., in the figure, the shear transformation for the top half is $(x,y) \rightarrow (x-\frac{5}{4} y, y)$). 
Since the smallest positive slope resulting from the first phase is 1,\footnote{Technically, the smallest possible positive slope that could be produced is  $1 - \varepsilon m$ (since $h \leq m$), but since we can pick $\varepsilon$ arbitrarily small, for our purposes we can omit this detail.}  the shear transformation guarantees that all resulting slopes are negative.
For the negative clauses we apply an equivalent shear transformation, to guarantee that all resulting slopes are positive.

Note that the shear transformations preserve segment disjointness, and thus also preserve the original nesting of the clauses.

\subsection{Variable gadgets and final point set}
\label{sec:variables}
\begin{figure}
    \centering
    \includegraphics{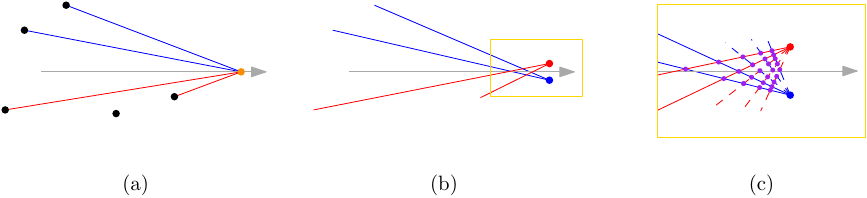}
    \caption{Variable gadgets. (a) Original variable point (orange). (b) Variable point replaced by two points (red and blue); all pairs of incident lines of different color cross.
    (c) Additional segments are added (dashed) to have exactly $m$ segment incident to each red and each blue point; in the final instance, one point is added at each intersection (purple points).}
    \label{fig:instance-example2-variables}
\end{figure}

Next we use a gadget similar to the one used by Megiddo and Tamir~\cite{MegiddoT82} to force a binary behavior on variables.  

Each variable point is replaced by two points.
One point is connected to the positive clauses that include the variable, and is placed slightly below the $x$-axis.
The other point is connected to the negative clauses where the variable appears, and is  located slightly above the $x$-axis.
The two variable points are placed so that the pairs of lines to clauses of different sign cross, see Figures~\ref{fig:instance-example2-variables}(a)--(b).
Following the figure, we will refer to these new variable points and their incident segments as \emph{blue} and \emph{red}.

In principle, each red or blue point is connected by a line segment only to those clause points that contain the variable (in the corresponding positive or negative form).
Next, we add additional segments to complete the number of incident segments at each red and at each blue point to $m$.
These additional segments are relatively short; they  only need to be long enough to intersect all segments incident to the corresponding point of different color, see Figure~\ref{fig:instance-example2-variables}(c).
We add all additional line segments also along lines coming out from the corresponding red or blue point.
In the case of the blue points, we use negative slopes slightly smaller than the slope of the steepest blue edge.
In the case of red points, we do the analog, but with  positive slopes. 

After this, each red and each blue point has   exactly $m$ incident segments, and the two groups of segments cross at exactly $m^2$ points.

We are now ready to describe the final instance, which is a set of points $S$, composed as follows:

\begin{itemize}
    \item All clause points; $m$ in total. 

    \item At each variable gadget we add one point for each intersection between a red and a blue segment (i.e., we add $m^2$ points for each variable---$nm^2$ in total).
\end{itemize}

In total, $S$ has $m+nm^2$ points.
Note that variable points are \emph{not} part of the instance.

The construction relies on a few properties that are easy to satisfy by applying suitable perturbations:
(i) the line through each red or blue segment contains the $m$ points from its variable gadget, and possibly also the corresponding clause point (if and only if the corresponding variable appears in the clause), but no other point from $S$;
(ii) the only lines containing three or more points from $S$ are the lines through a red segment or through a blue segment.
       
\begin{lemma}
Point set $S$ can be covered with $nm$ disjoint line segments if and only if the associated formula can be satisfied.
\end{lemma}

\begin{proof}
    The proof is analogous to the one by Megiddo and Tamir~\cite{MegiddoT82}, with the main difference that instead of lines,  we can use disjoint line segments to cover the points. For completeness, we include the argument next.
    
    Consider a variable gadget, with its red and blue points, and its $2m$ red/blue line segments.
    The key property of the variable gadget is that the $m^2$ points of a variable gadget can be covered with $m$ line segments if and only if the $m$ segments  all belong to the $m$ lines from the red point, or they all belong to the $m$ lines from the blue point. 
    Moreover, the $m$ line segments can be short enough as to be disjoint.
    Thus, with $nm$ disjoint segments picked in this way, all variable gadget points can be covered, and the choice of the segments used gives a True/False assignment for each variable. 
    
    It remains to analyze the $m$ clause points.
    If the whole set $S$ can be covered with $nm$ disjoint line segments, then the segments are precisely those optimally covering the variable gadget points. 
    The only possibility is to extend some of these segments to cover also clause points. 
    Since, by construction, clause points are collinear only with red/blue segments where the corresponding literal appears, if no extra segment is needed, then every clause point has at least one segment extending until it, which is equivalent to the literal being True, and the clause being satisfied.

    Conversely, if the formula can be satisfied, one can cover all points with $nm$ disjoint line segments by just picking the optimal line segments for each variable. For each clause, we pick one of the literals that is True, and extend the corresponding line segment from the variable gadget until it contains the clause point. 
    In this way, no extra segment is needed to cover all clause points.
    Note that, since the resulting segments follow the structure of the transformed drawing of the formula, they are disjoint.
\end{proof}

It remains to argue that the set of points $S$ can be constructed in polynomial time.

\begin{lemma}
Point set $S$ can be constructed in polynomial time.
\end{lemma}
\begin{proof}
    The drawings obtained in phases 1 and 2 can be built in polynomial time from the graph of the input formula, since all coordinates involved are rationals, with values that are polynomials in $n$ and $m$. 
    The coordinates of the points in the variable gadgets are intersections of lines, whose  equations---ignoring the perturbations needed---are also given by polynomials. 
    The perturbations needed to avoid unwanted collinearities are analog to the ones used by Megiddo and Tamir~\cite{MegiddoT82}.
    Their same argument, based on expressing the lines as lines that need to go through a few points while avoiding a polynomial-size collection of \emph{forbidden} points, applies directly to our construction, implying that all needed coordinates are polynomially bounded on $n$ and $m$.
\end{proof}

\begin{theorem}
    Given a set $S$ of points in the plane and a value $k$, the problem of determining if all points in $S$ can be covered by $k$ disjoint line segments is NP-complete.
\end{theorem}
\begin{proof}
    First observe that the problem is in NP, since given a set of line segments, it is straightforward to verify in polynomial time that they cover all points, and that they are disjoint.
    
    The previous results show that, given an instance of planar monotone 3-SAT with $n$ variables and $m$ clauses, in polynomial time we can construct a set of points $S$ such that $S$ can be covered by $k=nm$ disjoint line segments if and only if the given 3-SAT formula can be satisfied.
    Since, the set of points can be constructed in polynomial time, the theorem follows.
\end{proof}

\section{Covering points by guillotine cuts}

In this section we observe that the previous construction can be covered by a sequence of guillotine cuts, therefore proving that covering a set of points by guillotine cuts is also an NP-complete problem.

Consider a sequence of interior-disjoint  line segments $s_1, \dots, s_k$.
With some abuse of terminology, we also consider halflines or lines as (degenerate) segments.
Segment $s_i$ is a \emph{guillotine cut} if it is maximal interior-disjoint with respect to segments $s_1, \dots, s_{i-1}$, meaning that any other longer segment in the plane fully containing $s_i$ will intersect a previous segment. 
Note that here we only forbid \textit{proper intersections}, that is, intersections at points that are interior for both segments.
This implies that guillotine cuts are segments that start and end at previous segments, unless they are unbounded in one of the endpoints.

As the example in Figure~\ref{fig:example-problems} shows, covering a set of points by guillotine cuts may require more cuts than covering them by disjoint line segments.
That is because guillotine cuts are more restrictive than disjoint line segments, in the sense that not all sets of disjoint line segments can be completed to guillotine cuts. 
Nevertheless, we show next that the previous proof can be adapted so that it also shows NP-completeness for guillotine cuts.

\begin{theorem}
    Given a set $S$ of points in the plane and a value $k$, the problem of determining if all points in $S$ can be covered by $k$ guillotine cuts is NP-complete.
\end{theorem}

\begin{proof}
We aim to demonstrate that an optimal segment covering of the instance constructed in the previous section can be achieved using a sequence of guillotine cuts, where the number of cuts is equal to the number of segments, thus proving optimality.

We first focus on a simplified subset of the segments: the original red and blue segments described in the second paragraph of Section~\ref{sec:variables}. While the full construction uses $m$ segments per variable, we temporarily ignore this detail.
Refer to Figure~\ref{fig:instance-example-cuts}(a).

Consider an optimal covering of this simplified instance. We can cover these segments using a sequence of guillotine cuts defined by simply identifying the rightmost remaining segment, and applying a guillotine cut along the line supporting this segment.

By construction, all clause points at the same level share the same $y$-coordinate. When a cut is made along the rightmost segment's supporting line, all remaining points are necessarily left on the same side of the cutting line, ensuring the cut is a valid guillotine cut. 
We then proceed recursively on the sub-region containing all remaining points.
See Figure~\ref{fig:instance-example-cuts}(b).

Next, we account for the full construction, where each original red or blue segment is actually composed of $m$ individual segments. 
Since the supporting lines of the $m$ segments associated to a variable go through the same variable point, the cut identified above (which covered the conceptual segment) can be realized by applying $m$ sequential guillotine cuts in a right-to-left order, each covering one of the $m$ segments.

This cutting process uses exactly as many guillotine cuts as there are segments in the optimal solution.

Since any set of guillotine cuts can be shortened to form a set of disjoint segments, a covering by guillotine cuts must use at least as many cuts as an optimal segment covering. Because our construction achieves the minimum possible number of segments, it yields an optimal guillotine cut covering.

Therefore, the previous NP-completeness reduction (for disjoint segment covering) also applies directly to the problem of covering points by guillotine cuts.
\end{proof}

\begin{figure}
    \centering
    \includegraphics{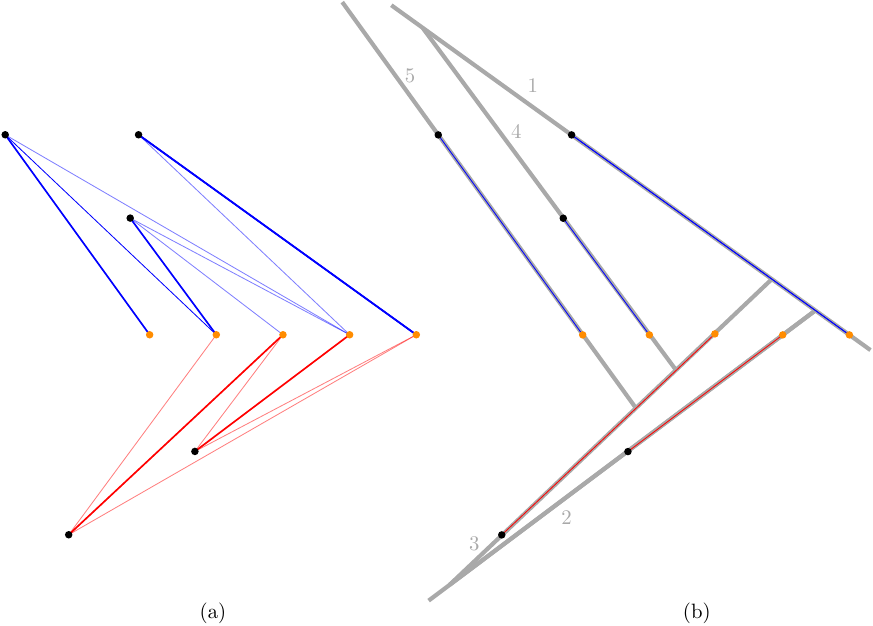}
    \caption{(a) Thicker segments indicate those present in an optimal solution (ignoring that each red/blue segment is actually composed of $m$ segments each).
    (b) Sequence of guillotine cuts associated with the solution (not showing the details of all $m$ cuts for each variable).}
    \label{fig:instance-example-cuts}
\end{figure}

\section{Conclusions}
We have shown that the problem of covering points by guillotine cuts is NP-complete.
In particular, first we revisited the problem of covering points by disjoint line segments, and presented a new hardness proof based on planar monotone 3-SAT.
A new proof seemed necessary as an intermediate step, given that the adaptation of the previous proof from disjoint line segments~\cite{grelier:LIPIcs.SoCG.2022.45} to guillotine cuts did not look straightforward.
Having established the---rather unsurprising---complexity of the problem justifies exploring other more interesting avenues.
In particular, future work will focus on understanding whether existing algorithms for variants of the line cover problem can be applied to the guillotine cut version, and on proving bounds between solution sizes of the different variants.

\bibliographystyle{plain} 

\bibliography{refs}

\end{document}